\documentclass{llncs}
\usepackage{amssymb}
\usepackage{amsmath}
\usepackage{wasysym}
\usepackage{textcomp}
\newtheorem{fact}{Fact}
\newtheorem{ex}{Example}
\newcommand{\dollar}[0]{\$}
\usepackage{comment}

\usepackage{makeidx}  
\begin{document}
\frontmatter          
\pagestyle{headings}  
\addtocmark{Real-Time Vector Automata} 
\mainmatter              
\title{Real-Time Vector Automata}
\titlerunning{Real-Time Vector Automata}  
%
\author{ \"{O}zlem Salehi\inst{1} \and Abuzer Yakary{\i}lmaz\inst{2,}\thanks{Yakary{\i}lmaz was partially
supported by FP7
FET-Open project QCS.}
 \and A. C. Cem Say \inst{1}}
\authorrunning{\"Ozlem Salehi et al.} 
%
%
\institute{Bo\v{g}azi\c{c}i University, Department of Computer Engineering, Bebek 34342 Istanbul, Turkey\\
\email{ozlem.salehi@boun.edu.tr},
\email{say@boun.edu.tr},
\and
University of Latvia, Faculty of Computing, Raina bulv. 19, R\={\i}ga, LV-1586, Latvia \\
\email{abuzer@lu.lv}
}

\maketitle              

\begin{abstract}
We study the computational power of real-time finite automata that have been augmented with a vector of dimension $k$, and programmed to multiply this vector at each step by an appropriately selected $k\times k$ matrix. Only one entry of the vector can be tested for equality to 1 at any time. Classes of languages recognized by deterministic, nondeterministic, and ``blind" versions of these machines are studied and compared with each other, and the associated classes for multicounter automata, automata with multiplication, and generalized finite automata.

\end{abstract}

\section{Introduction}
There have been numerous generalizations of the standard deterministic finite automaton model. In this paper,
we introduce the vector automaton, which is linked to many such generalizations like counter automata,
automata with multiplication, and generalized stochastic automata \cite{FMR68,Gr78,ISK76,Tu69}. A vector
automaton is a finite automaton endowed with a $k$-dimensional vector, and the capability of multiplying this
vector with an appropriately selected matrix at every computational step. Only one of the entries of the
vector can be tested for equality to 1 at any step. Since equipping these machines with a ``one-way" input
head, which is allowed to pause on some symbols during its left-to-right traversal of the input, would easily
make them Turing-equivalent, we focus on the case of real-time input, looking at the deterministic and
nondeterministic versions of the model. We make a distinction between general vector automata and ``blind"
ones, where the equality test can be  performed only at the end of the
computation. We examine the effects of restricting $k$ to 1, and the input alphabet to be unary. The related language classes are compared with each other, and classes associated with other models in the literature. The deterministic blind version of the model turns out to be equivalent to Turakainen's generalized stochastic automata in one language recognition mode, whereas real-time nondeterministic blind vector automata are shown to recognize some $\mathsf{NP}$-complete languages.

\section{Background}

\subsection{Notation}

Throughout the paper, the following notation will be used: $Q$ is the set of
states, where $q_0 \in Q$ denotes the initial state, $Q_a \subset Q$ denotes the
set of accept states, and $\Sigma$ is the input alphabet. An input string $w$ is placed between two endmarker symbols
on an infinite tape in the form $\cent w\dollar$. The set $\{\downarrow, \rightarrow , \leftarrow \}$ represents the possible head
directions. The tape head can stay in the same position ($\downarrow$), move one square to the right
($\rightarrow$), or move one square to the left ($\leftarrow$) in one step.

For a machine model $A$, $\mathfrak{L}(A)$
denotes the class of languages recognized by automata of type $A$.

Let $E^{i}_k(c)$ 
denote the matrix obtained by setting the $i$'th entry of the first column of the $k\times k$ identity matrix to $c$. For a vector $v$, the product $vE^{i}_k(c)$ is the vector obtained by adding $c$ times the
$i$'th entry of $v$ to the first entry when $i>1$, and the vector obtained by multiplying the first entry of $v$ by  $c$ when 
$i=1$.

\subsection{Machine definitions}

\subsubsection{Multicounter Automata.}

A \textit{real-time deterministic multicounter automaton} (rtD\textit{k}CA) \cite{FMR68} is a 5-tuple
\[ \mathcal{M}=(Q, \Sigma, \delta, q_0, Q_a). \]

The transition function $\delta$ of $\mathcal{M}$ is specified so that
$\delta(q,\sigma,\theta)=(q',c)$
means that $\mathcal{M}$
moves the head to the next symbol, switches to state $q'$, and updates its counters according to the list of increments represented by $c \in \{-1,0,1\}^k$,
if it reads symbol $\sigma \in \Sigma$, when  in state $q \in Q$, and
with the counter values having signs as described by $\theta \in \{0,\pm\}^k$. At the
beginning of the computation, the tape head is placed on the symbol $\cent$,
and the counters are set to 0. At the end of the computation, that is, after the right endmarker $\dollar$
has been scanned, the input is
accepted if $\mathcal{M}$ is in an accept state.

A \textit{real-time deterministic blind multicounter automaton} (rtD\textit{k}BCA) \cite{Gr78}
$\mathcal{M}$ is a rtD\textit{k}CA which can
check the value of its counters only at the end of the computation. Formally,
the transition function is now replaced by $\delta(q,\sigma)=(q',c).$ The input
is accepted at the end of the computation if $\mathcal{M}$ enters an accept
state, and all counter values are equal to 0.

\subsubsection{Finite Automata With Multiplication.}

A \textit{one-way deterministic finite automaton with multiplication} (1DFAM) \cite{ISK76} is a 6-tuple
\[\mathcal{M}=(Q,\Sigma,\delta,q_0,Q_a,\Gamma), \] where $\Gamma$ is a finite set of rational numbers
(multipliers). The transition function $\delta$ is defined as
$ \delta: Q \times \Sigma \times \Omega
\rightarrow Q\times \{\downarrow, \rightarrow\} \times \Gamma, $ where $\Omega=\{=,\neq\}. $
$\mathcal{M}$ has a register which can store any rational number, and is initially set to 1. Reading input
symbol
$\sigma$ in state $q$, $\mathcal{M}$ compares the current value of the register with 1, thereby calculating
the corresponding value $\omega \in \Omega$, and switches its state to $q'$, ``moves" its head in ``direction"
$d$, and multiplies the register by $\gamma$, in accordance with the transition function value
$\delta(q,\sigma,\omega)=(q',d,\gamma)$. The input string is accepted
if $\mathcal{M}$ enters an accept state with the register value equaling 1 after it scans the right endmarker
symbol.

A 1DFAM \textit{without equality} (1DFAMW) is a 1DFAM which can not check whether or not the register has
value 1 during computation. The transition function $\delta$ is replaced by $\delta(q,\sigma)=(q',d,\gamma)$.
The accept condition of the 1DFAMW is the same with the 1DFAM.

\subsubsection{Generalized Finite Automata.}

A \textit{generalized finite automaton} (GFA) \cite{Tu69} is  a
5-tuple
\[ \mathcal{G}=(Q,\Sigma,\{A_{\sigma \in \Sigma}\}, v_0, f), \]
where the $A_{\sigma \in \Sigma}$'s are $|Q|\times |Q|$ are real valued
transition matrices, and $v_0$ and $f$ are the real valued initial row vector and
final column vector, respectively. The acceptance value for an input string $w
\in \Sigma^*$ is defined as $f_{\mathcal{G}}(w)=v_oA_{w_1}\dots A_{w_{|w|}}f$.

A GFA whose components are restricted to be rational numbers is called a \textit{Turakainen finite automaton}
(TuFA) in \cite{Yak12}.

Let $\mathcal{G'}$ be a Turakainen finite automaton. Languages of the form
\[L=(\mathcal{G'},\mbox{=}\lambda) \equiv \{w\in \Sigma^* \mid f_{\mathcal{G'}}(w)=\lambda\}\]
for any $\lambda \in \mathbb{Q}$ constitute the
class $\mathsf{S}^=_{\mathbb{Q}}$.

\section{Vector Automata}

 A \textit{real-time deterministic vector automaton of dimension $k$} (rtDVA(\textit{k})) is a 6-tuple
\[\mathcal{V} =(Q,\Sigma,\delta,q_0,Q_a,v),\]
 where
$v$ is a $k$-dimensional initial row vector, and the
transition function $\delta$ is defined as
\[\delta: Q \times \Sigma \times \Omega \rightarrow Q\times S,\]
 where $S$ is the set of $k \times k$ rational-valued matrices, and $\Omega=\{=,\neq\}$, as in the definition
of 1DFAM's.

Specifically, $\delta(q,\sigma,\omega)=(q',M)$ means that when $\mathcal{V}$ is
in state $q$ reading symbol $\sigma \in \Sigma$, and the first entry of its
vector corresponds to $\omega \in \Omega$ (with $\omega$ having the value = if and only if this entry is equal
to 1),
$\mathcal{V}$ moves to state $q'$, multiplying its vector with the matrix $M
\in S$. As in the definition of 1DFAM's, $\omega$ is taken to be = if the first entry of the vector equals 1,
and $\neq$ otherwise. The string is accepted if
$\mathcal{V}$ enters
an accept state, and the first entry of the vector is 1, after processing the right end-marker symbol
$\dollar$.


\begin{remark}
The designer of the automaton is free to choose the initial setting $v$ of the vector.
\end{remark}

In the definition, it is stated that the machine can only check the first entry of the
vector for equality to 1. Sometimes we find it convenient to design programs that check for equality to some
number other than 1. One may also wish that it were possible to check  not the first, but some other entry
of the vector. In the following theorem, we show that we can assume our rtDVA(\textit{k})'s have that
flexibility. For
the purposes of that theorem, let a rtDVA(\textit{k})$^i_{c}$ be a machine similar to a rtDVA(\textit{k}), but
with a generalized
definition that enables it to check the $i$'th entry, for equality to the number $c$.

\begin{theorem}
 \textbf{i.} Given a \textup{rtDVA(\textit{k})}$^i_{1}$ recognizing a language $L$, one can construct a
rtDVA(\textit{k}) that recognizes
$L$.
\textbf{ii.} For any $c \in \mathbb{Q}$, given a \textup{rtDVA(\textit{k})}$^1_{c}$ recognizing a language
$L$, one can
construct a \textup{rtDVA($k$+1)} that recognizes $L$.
\end{theorem}

\begin{proof}

\textbf{i.} Suppose that we are given a rtDVA(\textit{k})$^i_{1}$
$\mathcal{V}  =(Q,\Sigma,\delta,q_0,Q_a,v)$. We will construct an equivalent
rtDVA(\textit{k}) $\mathcal{V'} =(Q,\Sigma,\delta',q_0,Q_a,v')$.
Let $J$ denote the matrix obtained from the $k\times k$ identity matrix by interchanging the first and $i$'th rows. 
We will use multiplications with $J$ repeatedly to swap the first and
$i$'th entries of the vector when it is time for that value to be checked, and then to restore the vector back
to its original order, so that the rest of the computation is not affected.
The initial vector of $\mathcal{V'}$ has to be a reordered version of $v$ to let the machine check the
correct entry at the first step, so $v'=vJ$. We update the individual transitions so that if
$\mathcal{V}$ has the move $\delta(q,\sigma,\omega)=(q',M)$, then $\mathcal{V'}$ has the move
$\delta'(q,\sigma,\omega)=(q',JMJ)$ for every $q \in Q$, $\sigma \in \Sigma$, and $\omega \in
\Omega$.

\textbf{ii.} Suppose that we are given a rtDVA(\textit{k})$^1_{c}$
$\mathcal{V}  =(Q,\Sigma,\delta,q_0,Q_a,v)$. We construct an equivalent rtDVA($k$+1) $\mathcal{V'}
=(Q,\Sigma,\delta',q_0,Q_a,v')$. The idea is to repeatedly subtract $(c-1)$ from the first entry of the
vector when it is time for that value to be checked, and then add $(c-1)$ to restore  the original
vector. We will use the additional entry (which will always equal 1 throughout the computation) in the vector
of $\mathcal{V'}$ to perform these additions and subtractions, as will be explained soon. Let $v''$ be a
$(k+1)$-dimensional
vector equaling $[v_1,v_2,...,v_k,1]$, where $v=[v_1,v_2,...,v_k]$. The
initial vector of $\mathcal{V'}$ has to be a modified version of $v''$ to accommodate the check for equality
to $1$ in the first step, so $v'=v''E^{k+1}_{k+1}(-c+1)$. For every individual transition
$\delta(q,\sigma,\omega)=(q',M)$ of $\mathcal{V}$, $\mathcal{V'}$ has the move
$\delta'(q,\sigma,\omega)=(q',E^{k+1}_{k+1}(c-1)NE^{k+1}_{k+1}(-c+1)) $, where the $(k+1) \times (k+1)$ matrix $N$ has
been obtained by adding a new row-column pair to $M$, i.e.
$N_{i,j}=M_{i,j} $ for $ i,j=1,...,k$, $N_{(k+1)j}=0
$ for $ j=1,...,k$, $N_{i(k+1)}=0$ for $ i=1,...,k$ and $N_{(k+1)(k+1)}=1$.

Note that when $c \neq 0$, there is an alternative method for constructing an equivalent rtDVA(\textit{k})
which does not require an extra entry in the vector, where the first entry is modified simply by repeated
multiplications with $E^{1}_k(1/c)$ and $E^{1}_k(c)$ when necessary.
\qed\end{proof}

We conclude this section with two examples that will familiarize us with the programming of
rtDVA(\textit{k})'s.

\begin{ex}
$\mathtt{UFIBONACCI}=\{a^n \mid $ n is a Fibonacci number$\}\in \mathfrak{L}\textup{(rtDVA(5))}$.
\end{ex}

\begin{proof}
We construct a rtDVA(5) $\mathcal{V}$ recognizing $\mathtt{UFIBONACCI}$ as follows: We let
the initial vector equal $[0,1,0,0,1]$. Reading each $a$, we multiply
the vector with the matrix $M_1$ if the first entry of the of the vector is
equal to 0, and with $M_2$ otherwise.

$$M_1=
\left [
\begin{array}{rrrrr}
0 & 0 & 0 & 0 & 0 \\
1 & 1 & 1 & 0 & 0\\
1 & 1 & 0 & 0 & 0 \\
-1 & 0 & 0 & 1 & 0 \\
-1 & 0 & 0 & 1 & 1
\end{array}
\right ]
M_2=
\left [
\begin{array}{rrrrr}
0&0&0&0&0 \\
1 & 1 & 0 & 0 & 0\\
0&0&1&0&0 \\
-1&0&0&1&0 \\
-1&0&0&1&1
\end{array} \right ]
.$$
After reading the $i$'th $a$, the
fourth entry of the vector equals $i$. The second and  third entries of
the vector hold consecutive Fibonacci numbers. The first entry is equal to 0
whenever $i$ equals the second entry, which triggers the next Fibonacci number to be computed and assigned to
the second entry in the  following step. Otherwise, the second and third entries remain unchanged until $i$
reaches the second entry. $\mathcal{V}$ accepts if the computation ends with the first
entry equaling 0, which occurs  if and only if the input length $n$ is a Fibonacci number.
\qed\end{proof}

\begin{theorem}\label{theorem:dva2}
$\mathtt{UGAUSS}=\{a^{n^2+n} \mid n \in \mathbb{N}\} \in \mathfrak{L}(\textup{rtDVA(2))}$.
\end{theorem}

\begin{proof}
We construct a rtDVA(2) $\mathcal{V}$ with initial vector $[1,1]$. If the input is the empty string, $\mathcal{V}$ accepts. Otherwise,
$\mathcal{V}$ increments the first entry of the vector by multiplying it by 2 on reading the first $a$ which is performed by multiplying the vector with the matrix $M_1=E^{1}_2(2)$.

$$M_1=
\left [
\begin{array}{rr}
2&0 \\
0&1
\end{array}
\right ]
$$
It then repeats the following procedure for the rest of the computation: Decrement the first entry of the vector by multiplying it by $1/2$ until it reaches one, while parallelly incrementing the second entry of the vector by multiplying it by 2 with the help of matrix $M_2$. The second entry stops increasing exactly when the first
counter reaches 1. Then the directions are swapped, with the second entry now being
decremented, and the first entry going up by multiplying the vector with the matrix $ M_3 $.
$$
M_2=
\left [
\begin{array}{rr}
\frac{1}{2}&0 \\
0 & 2
\end{array}
\right ]
M_3=
\left [
\begin{array}{rr}
2&0 \\
0 & \frac{1}{2}
\end{array}
\right ]
$$

When the second entry of the vector reaches 1, the first entry of the vector is multiplied by 2 one more time with the help of matrix $M_1$. Throughout this loop, the accept state is entered only when the
first entry of the vector is equal to 1.

Suppose that at some step, the value of the vector is $[1,2^c]$. If the input is sufficiently long,
$2c+2$ steps will pass before the first counter reaches 1 again, with the vector having the value
$[1,2^{c+1}]$. On an infinite sequence of $a$'s, the accept state will be entered after reading the second
$a$, and then again with intervals of $2c+2$ symbols between subsequent entrances, for $c=1,2,3...$. Doing
the sum, we conclude that strings of the form $a^{n^2+n}$, $n \in \mathbb{N}$, are accepted.
\qed\end{proof}

\section{Deterministic vector automata}

We start by specializing a fact stated by Ibarra et al. in \cite{ISK76} in the context of 1DFAM's to the case
of rtDVA(1)'s. For this purpose, we will use the following well-known
fact about counter machines.

\begin{fact}\cite{FMR68} \label{counterc}
 Given any $k$-counter automaton $A$ with the ability to alter the contents of
each counter independently by any integer between $+c$	 and $-c$ in a single
step (for some fixed integer $c$), one can effectively construct a $k$-counter automaton which can modify each
counter by at most one unit at every step, and which recognizes the same language as $A$ in precisely the same
number of steps.
\end{fact}

\begin{fact}\label{dva1-kca}
\textup{rtDVA(1)}'s are equivalent in language recognition power to real-time deterministic multicounter
automata
which can only check if all counters are equal to 0 simultaneously.
\end{fact}

\begin{proof}
 Let  us simulate a given rtDVA(1) $\mathcal{V}$ by a real-time
deterministic multicounter automaton $\mathcal{M}$. Let
$S=\{m_1,m_2,...,m_{t}\}$ be the set of numbers the single-entry ``vector" can be multiplied with
during the computation. Let $P=\{p_1,p_2,...,p_k\}$ be the set of prime factors of the denominators and the numerators of the numbers in $S$. $\mathcal{M}$ will have $k$ counters
$c_1,...,c_k$ to represent the current value
of the vector. When $\mathcal{V}$ multiplies the vector with $n_i=\frac{a}{b}$, where
$a=p_1^{x_1}p_2^{x_2}\dots p_k^{x_k}$ and $b=p_1^{y_1}p_2^{y_2}\dots p_k^{y_k}$, the
counters of  $\mathcal{M}$ are updated by the values $(x_1-y_1,x_2-y_2,...,x_k-y_k)$. As stated
in Fact \ref{counterc}, we can update the counter values by any integer between
$c$ and $-c$, where $c$ here is equal to the largest exponent in the prime
decomposition of the numbers in $S$. When $\mathcal{V}$  checks if the value of the
vector is equal to 1,  $\mathcal{M}$ checks if the current value of the
counters  is $(0,0,...,0)$, since the value of the vector is
equal to 1 exactly when all the counters are equal to 0.

For the other direction, we should simulate a rtD\textit{k}CA  $\mathcal{M}$ that can only check if all
counters are equal to 0 simultaneously with a
rtDVA(1) $\mathcal{V}$. For each counter $c_i$ of $\mathcal{V}$, we assign a distinct
prime number $p_i$ for $i=1,...,k$. We multiply the ``vector" with $p_i$ and
$\frac{1}{p_i}$, when the $i$'th counter $c_i$ is incremented and decremented, respectively. Whenever
$\mathcal{M}$ has all counters equal to 0, $\mathcal{V}$'s vector has value 1, so it can mimic $\mathcal{M}$
as required.
\qed\end{proof}

We now prove a fact about rtD\textit{k}CA's that will be helpful in the separation of the classes of languages
associated with these machines and rtDVA(1)'s.

\begin{theorem}\label{rtDkCA_unary}
$\mathtt{UGAUSS}=\{a^{n^2+n} \mid n \in \mathbb{N}\}\in \mathfrak{L}(\textup{rtD2CA})$.
\end{theorem}
\begin{proof}

We construct a real-time deterministic automaton $\mathcal{M}$ with two
counters recognizing $\mathtt{UGAUSS}$. The idea of the proof is the same with the proof of Theorem \ref{theorem:dva2}. If the input is the empty string, $\mathcal{M}$ accepts. Otherwise,
$\mathcal{M}$ increments the first counter on reading the first $a$. It then repeats the following procedure
for the rest of the computation: Decrement the first counter until it reaches zero, while parallelly
incrementing the second counter.  The second counter stops increasing exactly when the first counter reaches
0. The counters then swap directions, with the second counter now being decremented, and the first counter
going up. When the second counter reaches 0, the first counter is incremented one more time.

Throughout this loop, the accept state is entered only when the first counter is zero.

Suppose that at some step, the value of the counter pair is $(0,c)$. If the input is sufficiently long, $2c+2$
steps will pass before the first counter reaches zero again, with the pair having the value $(0,c+1)$. On an
infinite sequence of $a$'s, the accept state will be entered after reading the second $a$, and then again with
intervals of $2c+2$ symbols between subsequent entrances, for $c=1,2,3...$. Doing the sum, we conclude that
strings of the form $a^{n^2+n}$, $n \in \mathbb{N}$, are accepted.
\qed\end{proof}

For $k \geq 1$, let $\mathtt{LNG}_k=\{w \in \{a_0,a_1,...,a_k\}^{*} \mid
|w|_{a_0}=|w|_{a_1}=...=|w|_{a_k}\}$,
where $|w|_x$ denotes the number of occurrences of symbol $x$ in $w$.

\begin{fact}\label{fact:laing} \cite{La67}
$\mathtt{LNG}_k\in\mathfrak{L}$\textup{(rtD\textit{k}CA)}, and
$\mathtt{LNG}_k\notin\mathfrak{L}$\textup{(rtD($k$-1)CA)}, for every $k\geq 1$.
\end{fact}

\begin{fact}\label{fact:ISK}
\cite{ISK76}
\textup{1DFAM}'s can only recognize regular languages on unary alphabets.
\end{fact}

We are now able to state several new facts about the computational power of
rtDVA(\textit{k})'s:

\begin{theorem}
For any fixed $k>0$, $\mathfrak{L}$\textup{(rtDVA(1))} and $\mathfrak{L}$\textup{(rtD\textit{k}CA)}
are incomparable.
\end{theorem}

\begin{proof}
From Fact \ref{fact:laing}, we know that $\mathtt{LNG}_{k+1}$ can not be recognized by
any rtDkCA.
 We can construct a rtDVA(1) $\mathcal{V}$ recognizing $\mathtt{LNG}_{k+1}$ as follows:
 We
choose $k+1$ distinct prime numbers $p_1,...,p_k,p_{k+1}$, each corresponding to a different symbol $a_i$ in
the input alphabet, where
$i\in\{1,...,k+1\}$.
When it reads an $a_i$ with $i$ in that range, $\mathcal{V}$ multiplies its single-entry vector with $p_i$.
When it
reads an $a_{0}$, $\mathcal{V}$ multiplies the vector with $\frac{1}{p_1\cdot p_2\cdot...\cdot
p_kp_{k+1}}$.
The input string $w$ is accepted if the value of the vector is equal to 1 at the end of the
computation, which is the case if and only if $w\in\mathtt{LNG}_{k+1}$.
We conclude that $\mathtt{LNG}_{k+1} \in \mathfrak{L}$(rtDVA(1)).

From Theorem \ref{rtDkCA_unary}, we know that rtDkCA's can recognize some nonregular
languages on a unary alphabet. By Fact \ref{fact:ISK}, we know that rtDVA(1)'s, which  are additionally
restricted 1DFAM's, can only recognize regular languages in that case. Hence, we conclude that the
two models are incomparable.
\qed\end{proof}

\begin{theorem}\label{VA-kCA}
$\mathfrak{L}(\textup{rtDVA(1))} \subsetneq \bigcup_k \mathfrak{L}(\textup{rtD\textit{k}CA})$.
\end{theorem}

\begin{proof}
By the argument in the proof of Fact \ref{dva1-kca}, any rtDVA(1) can be simulated by a
rtD\textit{k}CA for some $k$. The inclusion is proper, since we know that a rtD2CA can
recognize a nonregular language on a unary alphabet (Theorem \ref{rtDkCA_unary}), a feat that is impossible
for rtDVA(1)'s by Fact \ref{fact:ISK}.
\qed\end{proof}

\begin{theorem}\label{theorem:lgrstar}
$\mathfrak{L}$\textup{(rtDVA(2))}
 $\nsubseteq \bigcup_k \mathfrak{L}\textup{(rtD\textit{k}CA)}$.
\end{theorem}

\begin{proof}
Let $\mathtt{GEQ}=\{a^mb^n | m\geq n \geq 1\}$, and let $\mathtt{GEQ}^*$ be the
Kleene closure of $\mathtt{GEQ}$. It is known that no rtDkCA can recognize $\mathtt{GEQ}^*$ for any $k$,
due to the inability of these machines to set a counter to 0 in a single step \cite{FMR67}.

 We will construct a rtDVA(2) $\mathcal{V}$ that recognizes $\mathtt{GEQ}^*$. The idea is to use the first
entry of the vector as a counter, and employ matrix multiplication to set this counter to 0 quickly when
needed.  $\mathcal{V}$ rejects strings that are not in the regular set $(a^{+}b^{+})^{*}$ easily. The vector
starts out as $[0,1]$. When it reads an $a$, $\mathcal{V}$ multiplies the vector with the ``incrementation"
matrix $M_a$ to increment the counter. When reading a $b$, $\mathcal{V}$ rejects if the first entry is zero,
since this indicates that there are more $b$'s than there were $a$'s in the preceding segment. Otherwise, it
multiplies the vector with the ``decrementation" matrix $M_b$.
 $$M_a=
\left [
\begin{array}{rr}
1&0 \\
1 & 1 \\
\end{array}
\right ]
M_b=
\left [
\begin{array}{rr}
1&0 \\
-1 & 1
\end{array}
\right ]
$$
When an $a$ is encountered immediately after a $b$, the counter has to be reset to 0, so the $M_a$ in the
processing of such $a$'s is preceded by the "reset" matrix $M_0$.
$$M_0=
\left [
\begin{array}{rr}
0 & 0 \\
1 & 1
\end{array}
\right ]
$$
$\mathcal{V}$ accepts if it reaches the end of the input without rejecting.
\qed\end{proof}

We are now able to compare the power of rtDVA(1)'s with their one-way versions, namely, the 1DFAM's of Ibarra
et al.
\cite{ISK76}

\begin{theorem}
$\mathfrak{L}$\textup{(rtDVA(1))} $\subsetneq$  $\mathfrak{L}$\textup{(1DFAM)}.
\end{theorem}

\begin{proof}
  We construct a 1DFAM $\mathcal{M}$ recognizing the language $\mathtt{GEQ}^*$ that we saw in the proof of
Theorem \ref{theorem:lgrstar}. $\mathcal{M}$ uses its register to simulate the counter of a one-way
single-counter automaton. When it reads an $a$, $\mathcal{M}$ multiplies the register by 2. When reading a new
$b$, $\mathcal{M}$ rejects if the register has value 1, and multiplies with  $\frac{1}{2}$ otherwise. When a
new block of $a$ is seen to start, $\mathcal{M}$ pauses its input head while repeatedly multiplying the
register with $\frac{1}{2}$ to set its value back to 1 before processing the new block. $\mathcal{M}$ accepts
if it has processed the whole input without rejecting.

By the already mentioned fact that no rtD\textit{k}CA for any $k$ can recognize $\mathtt{GEQ}^*$, and
Theorem
\ref{VA-kCA}, we conclude
that $\mathtt{GEQ}^*\notin \mathfrak{L}\textup{(rtDVA(1))}$.
\qed\end{proof}

The same reasoning also allows us to state
\begin{corollary}
 $\mathfrak{L}(\textup{rtDVA(1)}) \subsetneq \mathfrak{L}$\textup{(rtDVA(2))}.
\end{corollary}

Note that Fact \ref{fact:ISK} and Theorem \ref{theorem:dva2} let one conclude that \textup{rtDVA(2)}'s
outperform \textup{rtDVA(1)}'s when the input alphabet is unary.

It is easy to state the following simultaneous Turing machine time-space upper bound on the power of
deterministic real-time vector automata:
\begin{theorem}\label{theorem:rtDVAkinP}
$\bigcup_k \mathfrak{L}\textup{(rtDVA(\textit{k}))}\subseteq \mathsf{TISP}(n^3,n)$.
\end{theorem}
\begin{proof}
A Turing machine that multiplies the vector with the matrices corresponding to the transitions of a given
rtDVA(\textit{k}) requires only linear space, since the numbers in the vector can grow by at most a fixed
number of bits for each one of the $O(n)$ multiplications in the process. Using the primary-school algorithm
for multiplication, this takes $O(n^3)$ overall time.\qed
\end{proof}
If one gave the capability of one-way traversal of the input tape to vector automata of dimension larger than
$1$, one would gain a huge amount of computational power. Even with vectors of dimension 2, such machines can
simulate one-way 2-counter automata, and are therefore Turing equivalent \cite{ISK76}.  This is why we focus
on real-time vector
automata.

\section{Blind vector automata}

A \textit{real-time deterministic blind vector automaton} (rtDBVA(\textit{k})) is a
rtDVA(\textit{k}) which is not allowed to check the
entries of the vector until the end of the computation. Formally, a
rtDBVA(\textit{k}) is a 6-tuple
\[\mathcal{V} =(Q,\Sigma,\delta,q_0,Q_a,v),\]
 where the transition function $\delta$ is defined as
$ \delta: Q \times \Sigma \rightarrow Q\times S, $
with $S$ as defined earlier.
$\delta(q,\sigma)=(q',M)$ means that when $\mathcal{V}$ reads symbol $\sigma \in \Sigma$
in state $q$, it will move to state
$q'$, multiplying the vector  with the matrix $M \in S$. The acceptance condition is the same as for
rtDVA(k)'s.

\begin{remark}\label{remark:BDVA1}
Let us start by noting that $\mathfrak{L}(\textup{rtDBVA(1))} = \bigcup_k
\mathfrak{L}(\textup{rtD\textit{k}BCA})$, unlike the general case considered in Theorem \ref{VA-kCA}:  Since
blind counter automata only check if all counters are zero at the end, the reasoning of Fact \ref{dva1-kca} is
sufficient to conclude this.\end{remark}

\begin{theorem}\label{theorem:1WAYEQUALSRT}
$ \mathfrak{L}(\textup{rtDBVA(1)}) = \mathfrak{L}(\textup{1DFAM}W)$.
\end{theorem}
\begin{proof}
A rtDBVA(1) is clearly a 1DFAMW, so we look at the other direction of the equality. Given a 1DFAMW
$\mathcal{V}_1$, we wish to construct a rtDBVA(1) $\mathcal{V}_r$ which mimics $\mathcal{V}_1$, but without
spending more than one computational step on any symbol. When $\mathcal{V}_1$ scans a particular input symbol
$\sigma$ for the first time in a particular state $q$, whether it will ever leave this symbol, and if so,
after which sequence of moves, are determined by its program. This information can be precomputed for every
state/symbol pair by examining the transition function of $\mathcal{V}_1$. We program $\mathcal{V}_r$ so that
it rejects the input if it ever determines during computation that $\mathcal{V}_1$ would have entered an
infinite loop. Otherwise, upon seeing the simulated $\mathcal{V}_1$ moving on a symbol $\sigma$ while in state
$q$, $\mathcal{V}_r$ simply retrieves the aforementioned information from a lookup table, moves the head to the
right, entering the state that $\mathcal{V}_1$ would
enter when it moves off that $\sigma$, and multiplies its single-entry vector with the product of the
multipliers corresponding to the transitions $\mathcal{V}_1$ executes while the head is pausing  on $\sigma$.
\qed\end{proof}

We now give a full characterization of the class of languages
recognized by real-time deterministic blind vector automata.

\begin{theorem}\label{s=}
$\bigcup_k \mathfrak{L}(\textup{rtDBVA(\textit{k})})=\mathsf{S}^=_{\mathbb{Q}}$.
\end{theorem}

\begin{proof}
For any  language $L\in\mathsf{S}^=_{\mathbb{Q}}$, we can assume without loss of generality that $L=(\mathcal{G},\mbox{=}1)
$ \cite{Tu69} for some TuFA $\mathcal{G}$ with, say, $m$ states. Let us construct a
rtDBVA(\textit{k}) $\mathcal{V}$ simulating $\mathcal{G}$. We let $k=m$, so that the
vector is in $\mathbb{Q}^k$. The initial vector values of $\mathcal{V}$ and $\mathcal{G}$ are identical. $\mathcal{V}$ has only one state, and the vector is
multiplied with the corresponding transition matrix of $\mathcal{G}$ when an input
symbol is read. When processing the right endmarker, $\mathcal{V}$ multiplies the vector with a
matrix whose first column is the final vector $f$ of
$\mathcal{G}$.  $\mathcal{V}$ accepts input string $w$ if the first entry of
the vector is 1 at the end of the computation, which happens only if the acceptance value $f_{\mathcal{G}}(w)=1$.

For the other direction, let us simulate a rtDBVA(\textit{k}) $\mathcal{V}$ recognizing some language $L$ by a
TuFA $\mathcal{G}$. If $\mathcal{V}$
has $m$ states, then $\mathcal{G}$ will have  $km$ states. For any symbol $a$, the corresponding transition matrix $A$ is constructed as follows. View $A$ as being tiled to $m^2$ $k \times k$ submatrices
called $A_{i,j}$, for $i, j\in\{1,2,...,m\}$. If $\mathcal{V}$  moves from $q_i$ to $q_j$ by multiplying the vector
with the matrix $M_{i}$ when
reading symbol $a$, then  $A_{i,j}$ will be set to equal
$M_{i}$.  All remaining entries of $A$ are zeros. The
initial vector $v'$ of $\mathcal{G}$ will be a row vector with $km$ entries, viewed as being segmented to $m$ blocks of $k$ entries.  The first $k$
entries of $v'$, corresponding to the initial state of $\mathcal{V}$, will equal $v$, and the remaining entries of $v'$ will equal 0. The
$km$ entries of the final column vector $f$ of $\mathcal{G}$ will again consist of $m$ segments
corresponding  to the states of $\mathcal{V}$. The first entry of every such segment that corresponds to an
accept state of $\mathcal{V}$ will equal 1,  and all remaining entries will equal 0. $\mathcal{G}$ imitates the computation of $\mathcal{V}$ by keeping the current value of the vector of $\mathcal{V}$ at any  step within the segment that corresponds to $\mathcal{V}$'s current state in the vector representing the portion of $\mathcal{G}$'s own matrix multiplication up to that point. We therefore have that $L=(\mathcal{G},\mbox{=}1)$.
\qed\end{proof}

We can also give a characterization for the case where the alphabet is unary, thanks to the following fact,
which is implicit in the proof of Theorem 7 in \cite{Di77}:

\begin{fact}\label{fact:unaryreg}
  All languages on a unary alphabet in $\mathsf{S}^=_{\mathbb{Q}}$ are regular.
\end{fact}

We can say the following about the effect of increasing $k$ on the power of rtDBVA(\textit{k})'s:

\begin{theorem}
 $\mathfrak{L}$\textup{(rtDBVA(1))} $\subsetneq \mathfrak{L}$\textup{(rtDBVA(2))}.
\end{theorem}

\begin{proof}

Let us construct a rtDBVA(2) $\mathcal{V}$ recognizing the marked palindrome language
$\mathtt{MPAL}=\{wcw^r|w \in \{a,b\}^*\}$, where $w^r$ stands for the reverse of string $w$. We let
the initial vector equal $[0,1]$. While reading the input string, $\mathcal{V}$ first encodes the string $w$
in the first entry of the vector using the matrices $M_{a_1}$ and $M_{b_1}$.
$$M_{a_1} =
\left [
\begin{array} {rr}
10&~~0\\
1&1
\end{array}
\right ]
M_{b_1}=
\left [
\begin{array} {rr}
10&~~0\\
2&1
\end{array}
\right ]
$$

Each time it reads an $a$ and a $b$, $\mathcal{V}$ multiplies the vector with $M_{a_1}$ and $M_{b_1}$,
respectively. In the encoding, each $a$ is represented by an occurrence of the digit
1, and each $b$ is represented by a 2. Upon reading the symbol $c$, $\mathcal{V}$ finishes reading $w$ and
starts reading the rest of the string. $\mathcal{V}$ now makes a reverse encoding and multiplies
the vector with $M_{a_2}$ and $M_{b_2}$ each time it reads an $a$ and a $b$, respectively.
$$M_{a_2}=
\left [
\begin{array} {rr}
\frac{1}{10}&~~0\\
-\frac{1}{10}&1\\
\end{array}
\right ]
M_{b_2}=
\left [
\begin{array}{rr}
\frac{1}{10}&~~0\\
-\frac{2}{10}&1\\
\end{array}
\right ]
$$

When the computation ends, the first entry of the vector is equal to 0 iff the string read after the symbol
$c$ is the reverse of the string $w$ so that the input string is in $\mathtt{MPAL}$.

Now, we are going to prove that $\mathtt{MPAL} \notin \mathsf{2PFA}$, that is, the class of languages
accepted by two-way probabilistic finite automata with bounded error. Suppose for a contradiction that there
exists a two-way probabilistic finite automaton (2pfa) $\mathcal{M}$ recognizing $\mathtt{MPAL}$ with
bounded error. Then it is not hard show that $\mathtt{PAL}$ can be recognized by a 2pfa $ \mathcal{M'} $ such that $ \mathcal{M'} $ sees the input, say $ w $, as $ u=wcw $ and then executes $ \mathcal{M} $ on $ u $. Note that $ \mathcal{M} $ accepts $ u $ if and only if $ w $ is a member of $\mathtt{PAL}$.
Since $\mathtt{PAL}\notin\mathsf{2PFA}$ \cite{DS92}, we get a contradiction. Hence, we conclude that $\mathtt{MPAL}$ can not be in $\mathsf{2PFA}$.

 It is known \cite{Ra92} that $\mathsf{2PFA}$ includes all languages recognized by one-way deterministic blind multicounter automata, and  we already stated that \textup{rtDBVA(1)} and \textup{rtD\textit{k}BCA} are equivalent
 models in Remark \ref{remark:BDVA1}. Since $\mathtt{MPAL}\notin\mathsf{2PFA}$, $\mathtt{MPAL}$ cannot be in $\mathfrak{L}$\textup{(rtDBVA(1))}. Having proven that $\mathtt{MPAL} \in \mathfrak{L}$\textup{(rtDBVA(2))}, we conclude that $\mathfrak{L}$\textup{(rtDBVA(1))} $\subsetneq \mathfrak{L}$\textup{(rtDBVA(2))}.
\qed
\end{proof}

For an $m$-state rtDBVA(\textit{k}) $\mathcal{V}$, we define the \textit{size} of $\mathcal{V}$ to be the product $mk$. For all $i\geq1$, let $\mathfrak{L}$\textup{(rtDBVASIZE($i$))} denote the class of languages that are recognized by real-time deterministic blind vector automata whose size is $i$. We use the following fact to prove a language hierarchy  on this metric.

\begin{fact} \label{fact:rec} \cite{Di71} (Recurrence Theorem)
Let $\mathtt{L}$ be a language belonging to $S^=_{\mathbb{Q}}$ in the alphabet $\Sigma$. Then there exists a natural
number $n \geq 1$ such that for any words $x,y,z \in \Sigma^*$, if $yz, yxz,...,yx^{n-1}z \in \mathtt{L}$,
then $yx^mz \in \mathtt{L}$ for any $m \geq 0$.
\end{fact}

\begin{theorem}\label{thm:hier}
For every $i>1$, $\mathfrak{L}$\textup{(rtDBVASIZE($i-1$))}$\subsetneq\mathfrak{L}$\textup{(rtDBVASIZE($i$))}.
\end{theorem}

\begin{proof}
We first establish a hierarchy of complexity classes for TuFA's based on the number of states, and use this
fact to conclude the result.

It is obvious that the language $\mathtt{MOD_k}=\{a^{i} \mid i \neq 0 \mod k\}\in S^=_{\mathbb{Q}}$. We claim that any
TuFA $\mathcal{G}$ recognizing $\mathtt{MOD_k}$ should have at least $k$ states. Let $n$ be the number of states
of $\mathcal{G}$ and let us suppose that $n<k$. We are going to use Fact \ref{fact:rec} as follows: Let
$x=a$, $y=a$ and let $z$ be the empty string. Since the strings $a,a^2,...,a^{n-1}$ are in $\mathtt{MOD_k}$,
we see that the strings of the form $a^+$ are also in $\mathtt{MOD_k}$ and we get a
contradiction. Hence, we conclude that $n \geq k$ should hold, and that $\mathcal{G}$ should have at least $k$
states.

By Theorem $\ref{s=}$, there exists a real-time blind deterministic vector automaton with size $k$ (a rtDBVA(\textit{k}) with just one state) recognizing the same language. Suppose that there exists another real-time blind vector automaton $\mathcal{V}$ with size $k'$ such that $k'<k$. Then by Theorem $\ref{s=}$, there exists a TuFA with $k'$ states recognizing $\mathtt{MOD_k}$. Since we know that any TuFA recognizing $\mathtt{MOD_k}$ should have at least $k$ states, we get a contradiction.
\qed\end{proof}

\section{Nondeterministic  vector automata}

We now define the \textit{real-time nondeterministic vector automaton}
(rtNVA(\textit{k})) by adding the capability of making nondeterministic
choices to the rtDVA(\textit{k}). The transition function $\delta$ is now replaced by $\delta: Q \times
\Sigma \times \Omega \rightarrow \mathbb{P}(Q\times S) $, where $\mathbb{P}(A)$ denotes the
power set of the set $A$. We will also study blind versions of these machines: A \textit{real-time
nondeterministic blind vector automaton} (rtNBVA(\textit{k})) is just a rtNVA(\textit{k}) which does not check
the vector entries until the end of the computation.

We start by showing that it is highly likely that rtNVA(\textit{k})'s are more powerful than their
deterministic versions.

\begin{theorem}
If $\bigcup_k \mathfrak{L}(\textup{rtNVA(\textit{k})})=\bigcup_k \mathfrak{L}(\textup{rtDVA(\textit{k})})$, then $ \mathsf{P}=\mathsf{NP}$.
\end{theorem}

\begin{proof}
We construct a rtNBVA(3) $\mathcal{V}$ recognizing the $\mathsf{NP}$-complete language
$\mathtt{SUBSETSUM}$,
which is the collection of all strings of the form $t         \#  a_1\#...\# a_n\#$, such
that $t$ and the $a_i$'s are numbers in binary notation $(1 \leq i \leq n)$, and there
exists a set $I \subseteq \{1, . . . , n\}$ satisfying $\sum_{i \in
I}a_i=t$, where $n > 0$.
 The main idea of this construction is
that we can encode the numbers appearing in the input string to certain  entries of the vector, and perform
arithmetic on them, all in real time. We use a similar encoding given in \cite{Yak13A}. $\mathcal{V}$'s initial
vector is $[0,0,1]$. When scanning the
symbols of $t$, $\mathcal{V}$ multiplies the vector with the matrix $M_0$ (resp. $M_1$) for each scanned $0$
(resp. $1$).

$$M_0=
\left[
\begin{array}{rrr}
2&0&0 \\
0&1&0\\
0&0&1
\end{array}
\right ]
M_1=
\left[
\begin{array}{rrr}
2&0&0\\
0&1&0\\
1&0&1
\end{array}
\right ]
.$$
 When $\mathcal{V}$ finishes reading $t$, the vector equals
$[t,0,1]$. In the rest of the computation, $\mathcal{V}$ nondeterministically decides which $a_i$'s to
subtract from the second entry.
Each selected $a_i$ is encoded in a similar fashion to the fourth entry of the vector, using the matrices

$$N_0 =
\left [
\begin{array}{rrr}
1&0&0 \\
0&2&0 \\
0&0&1\\
\end{array}
\right ]
N_1 =
\left [
\begin{array}{rrr}
1&0&0 \\
0&2&0\\
0&1&1 \\
\end{array}
\right ]
.
$$
After encoding the first selected $a_i$, the vector  equals $[t,a_i,1]$. $\mathcal{V}$
subtracts the second entry from the first entry by multiplying the vector with the matrix
$E^{2}_3(-1)$. After this subtraction, the second entry is reinitialized to 0. $\mathcal{V}$ chooses another
$a_j$ if it wishes, and the same procedure is applied. At
the end of the input, $\mathcal{V}$ accepts if the first entry of the vector is equal to 0,  and rejects
otherwise.

If $\bigcup_k \mathfrak{L}$(rtNVA(\textit{k}))=$\bigcup_k \mathfrak{L}$(rtDVA(\textit{k})), then $\mathtt{SUBSETSUM}$
would be in $\mathsf{P}$ by Theorem \ref{theorem:rtDVAkinP}, and we would have to conclude that
$\mathsf{P}=\mathsf{NP}$.
\qed\end{proof}

When we restrict consideration to blind automata, we can prove the following unconditional separation between the deterministic and nondeterministic versions.

\begin{theorem}
$\mathfrak{L}(\textup{rtNBVA(2)})\nsubseteq\bigcup_k \mathfrak{L}(\textup{rtDBVA(\textit{k})})$.
\end{theorem}

\begin{proof}
Let us construct a rtNBVA(2) $\mathcal{V}$ recognizing the language $\mathtt{POW}=\{a^{k+2^k}\mid k>0\}$. The initial value of $\mathcal{V}$'s vector is $[1,1]$. $\mathcal{V}$'s computation consists of two stages. In the first stage, $\mathcal{V}$ doubles the value of the first entry for each $a$ that it scans, by multiplying  the vector
with the matrix $M_1$. At any step, $\mathcal{V}$ may nondeterministically decide to enter the second stage. In the second stage,  $\mathcal{V}$ decrements the first entry by 1, for each $a$ that is scanned, using the matrix $M_2$, and accepts if the first entry equals 0 at the end.
$$M_1=
\left [
\begin{array}{rr}
2 & 0
\\ 0 & 1
\end{array}
\right ]
 M_2=
 \left [
 \begin{array}{rr}
1 & 0
\\ -1 & 1
\end{array}
\right ]$$
If the input length is $n$, and if $\mathcal{V}$ decides to enter the second stage right after the $k$'th $a$, the vector value at the end of the
computation equals $[ 2^k - (n-k) ,  1 ]$. We see that $2^k - (n-k)=0$   if and only if $n=k+2^k$ for
some $k$.

Having proven that the nonregular language $\mathtt{POW}\in \mathfrak{L}$(rtNBVA(2)), we note that
$\mathtt{POW}$ can not be in $\bigcup_k \mathfrak{L}$(rtDBVA(\textit{k})), by Theorem \ref{s=}, and Fact
\ref{fact:unaryreg}.
\qed\end{proof}

\section{Open Questions}
\begin{itemize}

\item Can we show a hierarchy result similar to Theorem \ref{thm:hier} for general deterministic vector automata, or for nondeterministic vector automata?
\item Are general nondeterministic real-time vector automata more powerful than rtNBVA(\textit{k})'s?
\item Would properly defined bounded-error probabilistic versions of vector automata correspond to larger classes? Would quantum vector automata outperform the probabilistic ones?

\end{itemize}
\section*{Acknowledgements}
We thank Oscar Ibarra and Holger Petersen for their helpful answers to our questions.

\bibliographystyle{alpha}
 \bibliography{YakaryilmazSay}

\end{document}